\documentclass[conference]{IEEEtran}

\usepackage[table]{xcolor}
\usepackage{cite}
\usepackage[cmex10]{amsmath}
\usepackage{amssymb}
\usepackage{amsthm}
\usepackage{booktabs}
\usepackage{bm}
\usepackage{algorithm,algpseudocode}
\usepackage{multirow}

\usepackage[font=footnotesize,font+=sc]{caption}
\usepackage[font=footnotesize]{subfig}
\usepackage{filecontents}

 \newtheorem{theorem}{Theorem}
 
 \newtheorem{proposition}[theorem]{Proposition}
 
 \newtheorem{lemma}[theorem]{Lemma}
 
\theoremstyle{definition}

\theoremstyle{remark}

\newcommand{\cv}[1]{\mathbf{#1}}
\newcommand{\mc}[1]{\mathcal{#1}}

\newcommand{\Qin}{\mc Q_{\mathrm{in}}}
\newcommand{\Qout}{\mc Q_{\mathrm{out}}}

\title{Upper Bound Scalability on Achievable Rates of \\Batched Codes for Line Networks}
\author{\IEEEauthorblockN{Shenghao~Yang~and~Jie~Wang}%
\IEEEauthorblockA{The Chinese University of Hong Kong, Shenzhen}}%

\begin{document}

\maketitle

\begin{abstract}
The capacity of line networks with buffer size constraints is an open, but practically important problem. In this paper, the upper bound on the achievable rate of a class of codes, called batched codes, is studied for line networks. Batched codes enable a range of buffer size constraints, and are general enough to include special coding schemes studied in the literature for line networks. Existing works have characterized the achievable rates of batched codes for several classes of parameter sets, but leave the cut-set bound as the best existing general upper bound. In this paper, we provide upper bounds on the achievable rates of batched codes as functions of line network length for these parameter sets. Our upper bounds are tight in order of the network length compared with the existing achievability results.
\end{abstract}

\section{Introduction} 
The communication in a network from a source node to a destination
node may go through multiple hops, each of which
introduces errors. In this paper, we are interested in the problem that
when the intermediate nodes have buffer size constraints, how the communication rate scales with the number of hops.

In particular, we consider a line network of $L$ hops formed by a
sequence of nodes, where discrete memoryless channels (DMCs) exist
only between two adjacent nodes. We call the first node \emph{source node} and the last node \emph{destination node}. Except
for the source and destination nodes, all the other nodes, called
\emph{intermediate nodes}, have one incoming channel and one outgoing
channel. Each intermediate node has a buffer of $B$ bits to keep the
content used between different intermediate processing steps. There
are no other storage and computation constraints on the network nodes.

For some cases of the problem, the answers are known. When
the buffer size $B$ is allowed to increase with the block length at the source
node, the min-cut capacity can be achieved using hop-by-hop
decoding and re-encoding \cite{cover06}. When the zero-error capacity
of each channel is nonzero, using a constant buffer size $B$ can achieve the zero-error capacity for any value of $L$~\cite{Niesen2007}. 

In this paper, we focus on the DMCs in the line network with \emph{finite input and output alphabets} and \emph{$0$ zero-error capacity}. 
Note that for most common channel models, e.g., binary symmetric channels and erasure channels, the zero-error capacities are zero. 
When all cascaded channels are identical, Niesen, Fragouli, and Tuninetti~\cite{Niesen2007} showed that a class of codes with a constant buffer size $B$ can achieve rates
$\Omega(e^{-c L})$, where $c$ is a constant.
They also showed that if the buffer size $B$ is of order $\ln L$, any rate below the capacity of
the channel $Q$ can be achieved.  
Recently, Yang et al. \cite{yang17monograph,yang19isitline} showed that the end-to-end
throughput can be lower bounded by $\Omega(1/\ln L)$ using an intermediate node buffer size $O(\ln\ln L)$.
\footnote{In this paper,
  we say that $f(n)=\Omega(g(n))$ if there exists a real constant
  $c>0$ and there exists an integer constant $n_0\ge1$ such that
  $f(n)\ge c\cdot g(n)$ for every integer $n\ge n_0$; $f(n)=O(g(n))$
  if there exists a real constant $c>0$ and there exists an integer
  constant $n_0\ge1$ such that $f(n)\le c\cdot g(n)$ for every integer
  $n\ge n_0$; and $f(n)=\Theta(g(n))$ if both $f(n)=\Omega(g(n))$ and
  $f(n)=O(g(n))$ are satisfied.}

In contrast to these achievability results, min-cut is still the best upper bound. Characterizing a non-trivial, general upper bound for a line network with buffer size constraints could be difficult as hinted in \cite{vellambi11}. We relax the difficulty of the problem by asking the scalability of the upper bound with the network length $L$ for a class of codes, called \emph{batched  codes}.

Batched codes provide a general coding framework for line networks
with buffer size constraints, and include the codes studied in the
previous works \cite{Niesen2007,yang17monograph,yang19isitline} to show the achievability results as special cases.
A batched code has an outer code and an inner code. 
The outer code encodes the information messages into \emph{batches}, 
each of which is a sequence of coded symbols, 
while the inner code performs a general network coding for the symbols belonging to the same batch.  The inner code, comprising of \emph{recoding} at network nodes on each batch separately, should be designed for specific channels.  
Batched codes have been studied for designing efficient network coding for packet
erasure channels~(see, for example, \cite{chou03,Silva2009}), and practical designs have been provided \cite{yang11ac,yang14bats}.

The upper bound scalability on the achievable rates of batched codes
provides important guidance for us to design batched codes for large networks.
For example, we want to know whether the exponential decade of the achievable rate with $L$ is necessary for $B = O(1)$, and whether we can do better than $\Omega(1/\ln L)$ when $B=O(\ln\ln L)$. These questions are answered in this paper (see Table~\ref{tab:1}).
In particular, we show that when $N=O(1)$, which implies $M, B = O(1)$, the achievable rates must be exponential decade with $L$. When $N=O(1/\ln L)$ and $M=O(1)$, which implies $B = O(1/\ln L)$, the achievable rate is $O(1/\ln L)$. These upper bounds have the same order of $L$ as the previous achievability results, and hence, together, provide tight capacity scalability results of batched codes for these parameter sets. 

Our results are proved in a general setting of line networks where the DMC channels in the line network can be arbitrarily different except for a mild technical condition. The main technique of our converse is to separate the end-to-end transition matrix induced by the inner code as the linear combination of two parts, where one part captures the communication bottleneck in the line network and the other part can be simply upper bounded.

After introducing batched codes, we first use line networks of packet erasure channels to illustrate our main technique (Sec~\ref{sec:erasure}). We then generalize the results to a broad class of channels called \emph{canonical channels}, which include BSCs and BECs (Sec~\ref{sec:canonical}). Finally, we present a technique to solve line networks of general DMCs with zero-error capacity zero (Sec~\ref{sec:genc}).

\begin{table}[tb]
  \centering
    \caption{Summarization of the achievable rate scalability for the channels with the zero-error capacity $0$ using batched codes. Here, $c$ and $c'$ have constant values that do not change with $L$.}
    \label{tab:1}
    \subfloat[][lower bound]{
    \rowcolors{2}{white}{gray!30}
      \begin{tabular}{cccc}
        \toprule
        batch size $M$ & inner block-length $N$ & buffer size $B$ & lower bound \\
        \midrule
        $O(1)$ & $O(1)$ &  $O(1)$ & $\Omega(e^{-cL})$   \\
        $O(1)$ & $O(\ln L)$ & $O(\ln\ln L)$ & $\Omega(1/\ln L)$  \\
        $O(\ln L)$ & $O(\ln L)$ & $O(\ln L)$ & $\Omega(1)$  \\
        \bottomrule
      \end{tabular}
    }\\
    \subfloat[][upper bound]{
      \rowcolors{2}{white}{gray!30}
      \begin{tabular}{ccccc}
        \toprule
        batch size $M$ & inner block-length $N$  & buffer size $B$ & upper bound \\
        \midrule
        arbitrary & $O(1)$ &  $O(1)$ &  {$O(e^{-c'L})$} \\
        $O(1)$ & $O(\ln L)$ & $O(\ln L)$ & {$O(1/\ln L)$} \\
        $O(\ln L)$ & $O(\ln L)$ & $O(\ln L)$  & min-cut\\
        \bottomrule
      \end{tabular}
    }
\end{table}

\section{Line Networks and Batched Codes}
\label{sec:batched}

In this section, we describe a general line network model and introduce batched codes, which form a general coding framework for line networks with buffer size constraints.

\subsection{General Description}
\label{sec:line}

A line network of length $L$ is formed by a sequence of nodes $v_\ell$, $\ell=0,1,\ldots,L$, where communication links exist only between nodes $v_{\ell-1}$ and $v_\ell$ for $\ell=1,\ldots,L$. We assume that each link $(v_{\ell-1},v_\ell)$ is a discrete memoryless channel (DMC) with the transition matrix $Q_\ell$, where the input and output alphabets are $\Qin$ and $\Qout$, respectively, both finite.
We study the communication from the source node $v_0$ to the destination node $v_L$, where all the intermediate nodes $v_1,\ldots, v_{L-1}$ can help with the communication.

Let $K$, $n$ and $M$ be positive integers, and $\mc{A}$ and $\mc B$ be finite alphabets. A batched code has an outer code and an inner code described as follows.
The message of the source node is formed by $K$ symbols from $\mathcal{A}$.
The outer code of a batched code, performed at $v_0$, encodes the message and generates $n$ batches, each of which has $M$ symbols from $\mathcal{A}$.
Here $M$ is called the \emph{batch size}, and $n$ is called the \emph{outer blocklength}.

Let $N$ be a positive integer called the \emph{inner blocklength}.
The inner code of a batched code is performed on different batches separately, and includes the recoding operations at nodes $v_0,\ldots, v_{L-1}$:
\begin{itemize}
\item At the source node $v_0$ that generates the batches,  \emph{recoding} is performed on the original $M$ symbols of a batch to generate $N$ recoded symbols (in $\Qin$) to be transmitted on the outgoing links of the source node. 
\item At an intermediate network node $v$ that does not need to decode the input symbols, \emph{recoding} is performed on the received symbols (in $\Qout$) belonging to the same batch to generate $N$ recoded symbols (in $\Qin$) to be transmitted on the outgoing links of $v$. 
\end{itemize}
In general, the number of recoded symbols transmitted by different nodes can be different. Here we assume that they are all the same for the simplicity of the discussion.

\subsection{Recoding Formulations}

Let us formulate recoding for a generic batch $\cv X$. We denote by $\cv X[k]$ ($1\leq k\leq M$) the $k$th symbol in $\cv X$.
(Similar notations apply to other symbols of a sequence of entries.)
The recoding at the source node is a function $f:\mc A^M \rightarrow \Qin^N$. For $\ell = 1,\ldots,L$, denote by $\cv U_{\ell}$ and $\cv Y_{\ell}$ the input and output of $N$ uses of the link $(v_{\ell-1},v_{\ell})$, where $\cv U_1 = f(\cv X)$. Due to the memoryless of the channel, 
\begin{equation}\label{eq:ch}
  \Pr\{\cv Y_{\ell}=\cv y|\cv U_{\ell}=\cv u\} = Q_\ell^{\otimes N}(\cv y | \cv u) \triangleq \prod_{i=1}^N Q_\ell(\cv y[i]|\cv u[i]).
\end{equation}
The channel inputs $\cv U_{\ell}$, $\ell = 2,3,\ldots L-1$ can be formulated recursively.
Let $N'$ be an integer in $\{0,1,\ldots,N\}$ used to represent the input-output latency. 
For $i=0,1,\ldots,N+N'$, let $\cv B_{\ell}[i]$ be a random variable over the finite set $\mc B$ with $\cv B_{\ell}[0]$ a constant, which is used to represent the content in the buffer for the batch $\cv X$.
The recoding at $v_\ell$ is the function $\phi_\ell$ such that for $i=1,\ldots,N+N'$
\begin{equation}\label{eq:re}
  \left(\cv B_{\ell}[i],\cv U_{\ell+1}[i-N']\right) = \phi_{\ell}\left(\cv B_{\ell}[i-1],\cv Y_{\ell}[i]\right),
\end{equation}
where $\cv U_{\ell+1}[i]$ and $\cv Y_{\ell}[i]$ with
$i\notin\{1,\ldots,N\}$ are regarded as empty random variables.
In other words,
\begin{itemize}
\item For the first $N'$ received symbols, the recoding only updates its buffer content, but does not generate any channel inputs. 
\item After receiving $N'$ symbols, the recoding generates $N$ channel inputs. 
\end{itemize}
An inner code (or recoding) scheme is the specification of
$f$, $N$, $N'$ and $\{\phi_\ell\}$.

At the destination node, all received symbols (which may belong to different batches) are decoded jointly. The end-to-end transformation of a
batch is given by the transition matrix from $\cv X$ to $\cv Y_{L}$, which can be derived using \eqref{eq:ch} and \eqref{eq:re} recursively.
In general, the source recoding function $f$ and the intermediate recoding functions $\{\phi_\ell\}$ can be random. Let $F$ be the transition matrix from $\cv X$ to $\cv U_1$ and let $\Phi_\ell$ be the transition matrix from $\cv Y_\ell$ to $\cv U_{\ell+1}$.
We have the Markov chain
\begin{equation*}
  \cv X \rightarrow \cv U_1 \rightarrow \cv Y_1 \rightarrow \cdots \rightarrow \cv U_L \rightarrow \cv Y_L.
\end{equation*}
The end-to-end transition matrix from $\cv X$ to $\cv Y_{L}$ is
\begin{equation}\label{eq:tran}
  W_L \triangleq FQ_1^{\otimes N}\Phi_1Q_2^{\otimes N}\Phi_2\cdots Q_{L-1}^{\otimes N}\Phi_{L-1}Q_{L}^{\otimes N}.
\end{equation}

\subsection{Design Considerations}

The major parameters of a batched code include:
i) batch size $M$,
ii) inner blocklength $N$, and
iii) buffer size $B$ at the intermediate nodes. 
The buffer size $B = \log |\mc B|$ when $N'=0$, and $B = 2\log|\mc B|$ when $N'>0$.
For a given recoding scheme, the maximum achievable rate of the outer code is $\max_{p_{\cv X}} I(\cv X;\cv Y^{(L)})$ for $N$ channel uses. In other words, the design goal of a recoding scheme is to maximize
\begin{equation}\label{eq:recoding}
  C_L \triangleq \frac{1}{N}\max_{p_{\cv X}} I(\cv X;\cv Y_{L}) = \frac{1}{N}\max_{p_{\cv X}} I(p_{\cv X}, W_L)
\end{equation}
under certain constraints of $M$, $N$ and $B$ to be discussed later.
For a given recoding scheme, an outer code should be designed for the transition matrix $W_L$.
The optimal value of~\eqref{eq:recoding} is called the capacity of the line network with batched codes (under a certain constraint of $M$, $N$ and $B$), denoted as $C_L$.

By the convexity of mutual information for $W_L$ when $p_{\cv X}$ is fixed, we have the following proposition.

 \begin{proposition}\label{pro:deterministic}
   There exists a deterministic capacity achieving recoding scheme.
 \end{proposition}


\subsection{Capacity Scalability}

Under various constraints of $M$, $N$ and $B$, we study how the capacity of a line network with batched codes scales with the network length $L$.
Denote by
$C(Q_{\ell})$ and $C_0(Q_{\ell})$ the channel capacity and the
zero-error capacity of $Q_{\ell}$, respectively. Note that if
$C_0(Q_{\ell}) >0$ for any $\ell$, a constant rate can
be achieved for any network length $L$ using fixed $M$, $N$ and
$B$ (see also \cite{Niesen2007}).  
The same scalability result can be extended to
a line network with only a fixed number of DMCs $Q_{\ell}$ with
$C_0(Q_{\ell}) =0$. Henceforth in this paper, we consider the case that
$C_0(Q_{\ell}) =0$ for all $\ell$.


\subsubsection{$M=\Theta(N)$, $B=\Theta(N)$ and $N\rightarrow \infty$}

\emph{Decode-and-forward} is an optimal recoding scheme and achieves the
min-cut capacity $\min_{\ell=1}^L C(Q_{\ell})$ when i) $B$ is not
limited and ii) $N$ is allowed to be arbitrarily
large~\cite{cover06}.



\subsubsection{$M=\Theta(N)$, $B=\Theta(N)$ and $N=O(\ln L)$}

 As $N$ does not tend to infinity, the error probability at each intermediate node does not tend to zero if decode and forward is applied. 
When $Q_{\ell}$, $\ell=1,\ldots,L$ are identical, any constant rate below the channel capacity $C(Q_{\ell})$ 
can be achieved using batched codes with $n\rightarrow \infty$ \cite{Niesen2007}.

\subsubsection{$N=O(1)$}

When $N$ is a fixed number that does not change with
$L$, it is sufficient to consider a fixed $B$ and $M$. 
When $Q_{\ell}$, $\ell=1,\ldots,L$ are identical, $C_L$ tends to zero as $L\rightarrow\infty$ \cite{Niesen2007}. It was also shown that when $\Phi_{\ell}$, $\ell=1,\ldots,L-1$ are also identical, the maximum achievable rate converges to zero \emph{exponentially} fast.

When $N=O(1)$, the scalability of $C_L$ for general cases is still
open. For example, it is unknown whether this \emph{exponential
convergence} of the achievable rate still holds when channels and recoding functions at
intermediate nodes can be different.
In this paper, we will answer this question by a general upper bound that decreases exponentially in $L$.


\subsubsection{$M=O(1)$}

We are also interested in the case that $M$ is a relatively small,
fixed number that does not change with the network length $L$, so that
the major parameters of the outer code do not depend on the network
size. This may have certain advantages for the hardware implementation
of the outer code. 
It was shown in \cite{yang19isitline} that when $N=O(\ln L)$ and
$B = O(\ln\ln L)$, rate $\Omega(1/\ln L)$ can be achieved. Note that
$B = O(\ln N)$ is necessary when  a node needs at least to count how many
packets of a batch has been received.  In this paper, we will show that when
$N=O(\ln L)$ and $B = O(N)$, $C_L$ is $O(1/\ln L)$.




\section{Line Networks of Packet Erasure Channels}
\label{sec:erasure}

We first discuss a special case that the channels $\{Q_\ell\}$ are identical \emph{packet erasure channels} with transition matrix $Q_{\text{erasure}}$. Fix an alphabet $\mc Q^*$ with $|\mc Q^*|\geq 2$.  Suppose that the input alphabet $\Qin$  and the output alphabet $\Qout$ are both $\mc Q^* \cup \{0\}$ where $0\notin \mc Q^*$ is called the erasure. For each $x \in \mc Q^*$,
\begin{equation*}
  Q_{\text{erasure}}(y|x) =
  \begin{cases}
    1 - \epsilon & \text{if } y = x, \\
    \epsilon & \text{if } y = 0,
  \end{cases}
\end{equation*}
where $\epsilon$ is a constant value in $(0,1)$ called the erasure probability.
The input $0$ can be used to model the input when the channel is not used for transmission and we define $Q(0|0) = 1$. 

The relation between the input $X$ and output $Y$ of a packet erasure channel can be written as a function
$Y = XZ$, where $Z$ is a binary random variable independent of $X$ with $\Pr\{Z = 0\} = 1 - \Pr\{Z = 1\} = \epsilon$. In other words, $Z$ indicates whether the channel output is the erasure or not. With this formulation, we can write for $\ell=1,\ldots,L$ and $i=1,\ldots, N$,
\begin{equation*}
  \cv Y_{\ell}[i] = \cv U_{\ell}[i] \cv Z_{\ell}[i]
\end{equation*}
where $\cv Z_{\ell}[i]$ are independent binary random variables with $\Pr\{\cv Z_{\ell}[i] = 0\} = \epsilon$. 

The main idea of our converse is that the worst link in a line network restricts the capacity. We define the event $E_0$ to capture the communication bottleneck
\begin{equation*}
  E_0 = \cup_{\ell=1}^L \{\cv Z_{\ell} = 0\} = \left\{\lor_{\ell=1}^L(\cv Z_{\ell} =  0)\right\},
\end{equation*}
where $\cv Z_{\ell}=0$ means $\cv Z_{\ell}[i]=0$ for all $i$.
In other words, $E_0$ is the event that for at least one link, all the $N$ uses of the channel for a batch are erasures. Define $W_L^{(0)}$ and $W_L^{(1)}$ as the transition matrix from $\mc A^M$ to $\Qout^N$ such that
\begin{IEEEeqnarray*}{rCl}
  W_L^{(0)}(\cv y|\cv x) & = & \Pr\{\cv Y_L=\cv y|\cv X=\cv x, E_0\}, \\
  W_L^{(1)}(\cv y|\cv x) & = & \Pr\{\cv Y_L=\cv y|\cv X=\cv x, \overline{E_0}\},
\end{IEEEeqnarray*}
where $\overline{E_0} = \left\{\land_{\ell=1}^L(\cv Z_{\ell} \neq 0)\right\}$. As $\cv X$ and $\cv Z_{\ell}$, $\ell=1,\ldots, L$ are independent, we have 
\begin{equation*}
  W_L = W_L^{(0)}p_0 + W_L^{(1)}p_1,
\end{equation*}
where $p_0 = \Pr\{E_0\}$ and $p_1 = 1-p_0$. As $I(p_{\cv X}, W_L)$ is
a convex function of $W_L$ when $p_{\cv X}$ is fixed, we obtain
\begin{equation*}
  I(p_{\cv X}, W_L) \leq p_0I(p_{\cv X}, W_L^{(0)})  + p_1I(p_{\cv X}, W_L^{(1)}).
\end{equation*}

\begin{lemma}\label{lemma:e}
  For a line network of identical packet erasure channels, $I(p_{\cv X}, W_L^{(0)}) = 0$.
\end{lemma}
\begin{proof}
  Denote by $P$ the (joint) probability mass function of the random variables we have defined for the batch codes. To prove the lemma, we only need to show for all $\cv x \in \mc A^M$ and $\cv y_L \in \Qout^N$,
  \begin{equation}\label{eq:csd}
    P(\cv y_L, \cv x, E_0) = P(\cv x) P(\cv y_L, E_0),
  \end{equation}
  which implies $I(p_{\cv X}, W_L^{(0)}) = 0$.

  Define a sequence of events for $\ell = 1,\ldots, L$,
  \begin{equation*}
    \overline{E_0^{(\ell)}} = \{\cv Z_{\ell'} \neq 0, \ell'>\ell\}.
  \end{equation*}
  We have $\{\cv Z_{\ell}=0\} \cap \overline{E_{0}^{(\ell)}}$, $\ell=1,\ldots,L$ are disjoint and $E_0 =  \cup_{\ell=1}^L \left[\{\cv Z_{\ell}=0\} \cap  \overline{E_0^{(\ell)}}\right]$.
  Therefore, 
  \begin{IEEEeqnarray*}{rCl}
    \IEEEeqnarraymulticol{3}{l}{P(\cv y_L, \cv x, E_0)} \\
    & = & \sum_{\ell} P\left(\cv y_L, \cv x, \cv Z_{\ell}=0, \overline{E_0^{(\ell)}}\right) \\
    & = & \sum_{\ell} \sum_{\cv y_{\ell}} \sum_{\cv u_{\ell}} P\left(\cv y_L, \cv x, \cv y_{\ell},\cv u_{\ell}, \cv Z_{\ell}=0, \overline{E_0^{(\ell)}}\right) \\
    & = & \sum_{\ell} \sum_{\cv y_{\ell}} P(\cv y_L, \overline{E_0^{(\ell)}}\Big| \cv y_{\ell}) \sum_{\cv u_{\ell}} P(\cv x, \cv u_\ell, \cv y_{\ell},\cv Z_{\ell}=0).
  \end{IEEEeqnarray*}
  We have 
  \begin{IEEEeqnarray*}{rCl}
    \IEEEeqnarraymulticol{3}{l}{\sum_{\cv u_{\ell}} P(\cv x, \cv u_\ell, \cv y_{\ell},\cv Z_{\ell}=0)} \\
    & = & \sum_{\cv u_{\ell}} P(\cv x, \cv u_\ell) P(\cv Z_{\ell}=0) P(\cv y_{\ell}|\cv u_{\ell},\cv Z_{\ell}=0) \\
    & = & \sum_{\cv u_{\ell}} P(\cv x, \cv u_\ell) P(\cv Z_{\ell}=0) P(\cv y_{\ell}|\cv Z_{\ell}=0) \\
    & = & P(\cv x) P(\cv y_{\ell},\cv Z_{\ell}=0) 
  \end{IEEEeqnarray*}
  where $P(\cv y_{\ell}|\cv u_{\ell},\cv Z_{\ell}=0) = P(\cv y_{\ell}|\cv Z_{\ell}=0)$ follows that $\cv Y_\ell = 0$ as $\cv Z_{\ell} = 0$. 
  Hence, we obtain
  \begin{IEEEeqnarray*}{rCl}
    P(\cv y_L, \cv x, E_0) & = & P(\cv x)  \sum_{\ell} \sum_{\cv y_{\ell}} P\left(\cv y_L, \overline{E_0^{(\ell)}}\Big | \cv y_{\ell}\right) P(\cv y_{\ell},\cv Z_{\ell}=0).
  \end{IEEEeqnarray*}

  Similarly, we have
  \begin{IEEEeqnarray*}{rCl}
    P(\cv y_L, E_0) & = & \sum_{\ell} \sum_{\cv y_{\ell}} P\left(\cv y_L, \overline{E_0^{(\ell)}}\Big| \cv y_{\ell}\right) P(\cv y_{\ell},\cv Z_{\ell}=0).
  \end{IEEEeqnarray*}
  Therefore, we show \eqref{eq:csd}.
\end{proof}

As $p_1 = (1-\epsilon^N)^L$ and
\begin{equation*}
  I(p_{\cv X}, W_L^{(1)})  \leq  \min\{ M \ln|\mc A|, N\ln |\Qout| \},
\end{equation*}
we have
\begin{equation}\label{eq:pec}
  C_L  \leq  \frac{(1-\epsilon^N)^L}{N} \min\{ M \ln|\mc A|, N\ln |\Qout| \}.
\end{equation}

\begin{theorem}
  For a line network of length $L$ of packet erasure channels with erasure probability $\epsilon$,
  \begin{enumerate}
\item When $N=O(1)$, $C_L = O((1-\epsilon^N)^L)$.
\item When $M=O(1)$ and $N=\Theta(\ln L)$, $C_L = O(1/\ln L)$. 
\item When $M=\Omega(\ln L)$ and $N=\Omega(\ln L)$, $C_L = O(1)$.
\end{enumerate}
\end{theorem}
\begin{proof}
  The theorem can be proved by substituting $M$ and $N$ in each case into \eqref{eq:pec}.
\end{proof}

\section{Converse for General Channels}

Consider a generic channel $Q:\Qin \to \Qout$. The relation between the input
$X$ and output $Y$ of $Q$ can be modeled as a function $\alpha$ (see \cite[Section~7.1]{yeung08}):
\begin{equation}\label{eq:ssd} 
  Y = \alpha(X,Z=(Z_x, x\in \Qin)) =  \sum_{x\in\Qin}\bm 1\{X=x\}Z_x,
\end{equation}
where $\bm 1$ denotes the indicator function, and $Z_x, x\in \Qin$ are independent random variables with alphabet $\Qout$ define as
\begin{equation}\label{eq:z}
 \Pr\{Z_x = y\} = Q(y|x).
\end{equation}
For $N$ uses of the channel $Q$, we can write
\begin{equation}\label{eq:alphaN}
  \cv Y = \alpha^{(N)}(\cv U, \cv Z),
\end{equation}
where $\cv Y[i] = \alpha(\cv U[i],\cv Z[i])$.

In this section, we consider general DMCs $Q_{\ell}$ for all $\ell$, which can be modeled as the function $\alpha_{\ell}$.
With the above formulation, we can write for $\ell=1,\ldots,L$, 
\begin{equation}\label{eq:chz}
  \cv Y_{\ell} = \alpha_{\ell}^{(N)}(\cv U_{\ell}, \cv Z_{\ell}).
\end{equation}

\subsection{Canonical Channels}
\label{sec:canonical}

For $0<\varepsilon\leq 1$, we call a channel $Q:\Qin \to \Qout$ an \emph{$\varepsilon$-canonical channel} if there exists $y_0 \in \Qout$ such that for every $x \in \Qin$, $Q(y_0|x)\geq \varepsilon$. The packet erasure channel, BSC and BEC are all canonical channels. Note that a canonical $Q$ has $C_0(Q)=0$.
We first consider the case that the channels $\{Q_{\ell}\}$ are all $\varepsilon$-canonical channels.
Define the event
\begin{equation*}
  E_0 = \left\{\lor_{\ell=1}^L(\cv Z_{\ell} = y_0)\right\},
\end{equation*}
where $\cv Z_{\ell} = y_0$ means $(\cv Z_{\ell}[i])_x = y_0$ for all $i$ and $x$.
The event $E_0$ means that there exists one link of the network such that all uses of the channel for transmitting a batch have the same output $y_0$.
Similar to the discussion in Section~\ref{sec:erasure}, the transition matrix $W_L$ can be expressed as
\begin{equation*}
  W_L=W_L^{(0)}p_0 + W_L^{(1)}p_1,
\end{equation*}
where $p_0=\text{Pr}\{E_0\}, p_1=\text{Pr}\{\overline{E_0}\}$, 
and
\begin{align*}
W_L^{(0)}(\cv y\mid\cv x)&=\text{Pr}\{\cv Y_L=\cv y\mid\cv X=\cv x, E_0\},\\
W_L^{(1)}(\cv y\mid\cv x)&=\text{Pr}\{\cv Y_L=\cv y\mid\cv X=\cv x, \overline{E_0}\}.
\end{align*}
Hence,
\begin{equation*}
  I(p_{\cv X}, W_L) \leq p_0I(p_{\cv X}, W_L^{(0)})  + p_1I(p_{\cv X}, W_L^{(1)}).
\end{equation*}

\begin{lemma}\label{lemma:2}
  When $Q_{\ell}$, $\ell = 1,\ldots, L$ are all $\varepsilon$-canonical channels,
  $p_1 \leq (1-\varepsilon^{|\Qin| N})^L$.
\end{lemma}
\begin{proof} We write
  \begin{align*}
  p_1&=
  \prod_{\ell=1}^L\bigg[
  1 - \prod_{i\in\{1,\ldots,N\}} \prod_{x\in\Qin}\text{Pr}((\cv Z_{\ell}[i])_x=y_0)
  \bigg]\\
  &= \prod_{\ell=1}^L\bigg[
  1 - \prod_{i\in\{1,\ldots,N\}}\prod_{x\in\Qin}Q_{\ell}(y_0| x)
  \bigg] \leq  (1-\varepsilon^{|\Qin|N})^L,
  \end{align*}
  where the second equality follows from \eqref{eq:z}.
\end{proof}

\begin{lemma}\label{lemma:3}
  For a line network of length $L$ of $\varepsilon$-canonical channels, $I(p_{\cv X}, W_L^{(0)}) = 0$.
\end{lemma}
\begin{proof}
  Similar as the proof of Lemma~\ref{lemma:e}, we have
  \begin{IEEEeqnarray*}{rCl}
    P(\cv y_L,\cv x,E_0) & = & \sum_{\ell} \sum_{\cv y_{\ell}} P(\cv y_L, \overline{E_0^{(\ell)}}\Big| \cv y_{\ell}) \\
    & & \sum_{\cv u_{\ell}} P(\cv x, \cv u_\ell) P(\cv Z_{\ell}=y_0) P(\cv y_{\ell}|\cv u_{\ell},\cv Z_{\ell}=y_0).
  \end{IEEEeqnarray*}
  By \eqref{eq:chz},  given $\cv Z_{\ell}=y_0$,
  \begin{IEEEeqnarray*}{rCl}
    \cv Y_{\ell} & = & \alpha^{(N)}(\cv U_{\ell}, \cv Z_{\ell}=y_0) = y_0,
  \end{IEEEeqnarray*}
and hence $P(\cv y_{\ell}\mid\cv u_{\ell},\cv Z_{\ell}=y_0) = P(\cv y_{\ell}\mid\cv Z_{\ell}=y_0)$.
Following the same argument as in Lemma~\ref{lemma:e},
\[
P(\cv y_L, \cv x, E_0) = P(\cv x) P(\cv y_L, E_0),
\]
which implies $I(p_{\cv X}, W_L^{(0)}) = 0$.
\end{proof}

Combining both Lemma~\ref{lemma:2} and Lemma~\ref{lemma:3}, we can assert that 
\[
C_L\leq \frac{(1-\varepsilon^{|\Qin| N})^L}{N} \min\{ M \ln|\mc A|, N\ln |\Qout| \},
\]
which implies the following theorem:
\begin{theorem}
  For a length-$L$ line network of $\varepsilon$-canonical channels with finite input and output alphabets,
  \begin{enumerate}
\item When $N=O(1)$, $C_L = O((1-\varepsilon^{|\Qin| N})^L)$.
\item When $M=O(1)$ and $N=\Theta(\ln L)$, $C_L = O(1/\ln L)$. 
\item When $M=\Omega(\ln L)$ and $N=\Omega(\ln L)$, $C_L = O(1)$.
\end{enumerate}
\end{theorem}

\subsection{General Channels}
\label{sec:genc}

Consider a channel $Q:\Qin \to \Qout$ with $C_0(Q)=0$, modeled as in \eqref{eq:ssd}-\eqref{eq:alphaN}.
Denote by $\varepsilon_Q$ the maximum value such that for any $x, x'\in \Qin$, there exists $y\in \Qout$ such that $Q(y|x)\geq \varepsilon_{Q}$ and $Q(y|x')\geq \varepsilon_{Q}$. Note that $C_0(Q)=0$ if and only if $\varepsilon_Q>0$.

\begin{lemma}\label{zero-error}
  For a channel $Q:\Qin \to \Qout$ with $C_0(Q)=0$, and any non-empty $\mc A\subseteq\Qin^N$,
  there exist $\cv z=((\cv z[i])_x \in \Qout, x\in \Qin, i=1,\ldots,N)$ and a subset $\mc B\subseteq\Qout^N$ with $|\mc B|\le \lceil |\mc A|/2 \rceil$ such that $\alpha^{(N)}(\cv x,\cv z)\in \mc B$ for any $\cv x\in \mc A$ and $\Pr\{\cv Z = \cv z\} \geq \varepsilon_{Q}^{|\Qin|N}$.
\end{lemma}
\begin{proof}
  The sequences in $\mc A$ can be put into $\lceil |\mc A|/2 \rceil$ pairs. For each pair $\cv x$ and $\cv x'$, there exists $\cv y$ such that for each $i=1,\ldots,N$, $Q(\cv y[i]|\cv x[i])\geq \varepsilon_{Q}$ and $Q(\cv y[i]|\cv x'[i]) \geq \varepsilon_{Q}$. Let $(\cv z[i])_{\cv x[i]}=\cv y[i]$ and $(\cv z[i])_{\cv x'[i]}=\cv y[i]$. After going through all the $\lceil |\mc A|/2 \rceil$ pairs, 
  let $\mc B$ be the collection of all $\cv y$, which satisfies $|\mc B|\le \lceil |\mc A|/2 \rceil$. For all $(\cv z[i])_x$ that have not been assigned, let $(\cv z[i])_x=y$ such that $Q(y|x)\geq \varepsilon_Q$. Hence $\Pr\{\cv Z = \cv z\} = \Pr\{(\cv Z[i])_x = (\cv z[i])_x, x\in \Qin, i=1,\ldots,N\} \geq \varepsilon_Q^{|\Qin|N}$.
\end{proof}

Assume that $L = L'K$, where $L'$ and $K$ are integers. 
As a result, the end-to-end transition matrix $W_L$ can be written as
\begin{equation*}
  W_L = F G_1 \Phi_{K} G_2 \Phi_{2K}\cdots G_{L'},
\end{equation*}
where for $i=1,\ldots,L'$, 
\begin{equation*}
  G_i = Q_{K(i-1)+1}^{\otimes N}\Phi_{K(i-1)+1} \cdots \Phi_{Ki-1}Q_{Ki}^{\otimes N}.
\end{equation*}
The length-$L$ network can be regarded as a length-$L'$ network of channels $G_i$, $i=1,\ldots, L'$.
Because of proposition~\ref{pro:deterministic}, without loss of optimality, we assume a deterministic recoding scheme, i.e., $F, \Phi_{\ell}$ are deterministic transition matrices.
The input $\cv X$ and output $\cv Y$ of $G_i$ can be written as a function
\begin{equation*}
  \cv Y = \alpha_{G_i}(\cv X, \cv Z_{\ell}, \ell= K(i-1)+1, \ldots, Ki),
\end{equation*}
where $\alpha_{G_i}$ can be determined recursively by $F, \{\Phi_{\ell}\}$ and \eqref{eq:alphaN}.

 When $K \geq  N \log_2|\Qin|$ and $\varepsilon_{Q_{\ell}} \geq \varepsilon$ for all $\ell$,
applying Lemma~\ref{zero-error} inductively, we know that there exists $\cv y_i$ and $\{\cv z_{\ell}\}$ such that 
$\alpha_{G_i}(\cv x, \cv z_{\ell}, \ell= K(i-1)+1, \ldots, Ki)=\cv y_i$ for all $\cv x\in \Qin^N$, and
\begin{equation*}
  \Pr\{\cv Z_{\ell} = \cv z_{\ell}, \ell= K(i-1)+1, \ldots, K\} \geq \varepsilon_Q^{|\Qin|NK}.
\end{equation*}
For $i = 1,\ldots, L'$, define events
\begin{equation*}
  E_i = \{\cv Z_{\ell} = \cv z_{\ell}, \ell= K(i-1)+1, \ldots, K\}.
\end{equation*}
Define the event 
\begin{equation*}
  E_0 = \left\{\lor_{i=1}^{L'}E_i\right\}.
\end{equation*}
Performing the similar analysis as in Section~\ref{sec:canonical} for the length-$L'$ network of channels $G_i$, $i=1,\ldots, L'$ with $E_0$ defined above, we obtain
\begin{IEEEeqnarray}{C}\label{eq:main}
C_L\leq \frac{(1-\varepsilon^{NK|\Qin|})^{L/K}}{N} \min\{ M \ln|\mc A|, N\ln |\Qout| \},\IEEEeqnarraynumspace
\end{IEEEeqnarray}
and hence the following result holds:

\begin{theorem}
  For a length-$L$ line network of channels $Q_{\ell}$ with finite input and output alphabets and  $\varepsilon_{Q_{\ell}}\geq \varepsilon>0$ for all $\ell$,  
\begin{enumerate}
\item When $N=O(1)$, $C_L = O((1-\varepsilon')^L)$ for certain $\varepsilon'\in (0,1)$.
\item When $M=O(1)$ and $N=\Theta(\ln L)$, $C_L = O(1/\ln L)$. 
\item When $M=\Omega(\ln L)$ and $N=\Omega(\ln L)$, $C_L = O(1)$.
\end{enumerate}
\end{theorem}

\section{Concluding Remarks}

This paper characterized the tight capacity upper bound of batched codes for line networks when the channels have finite alphabets and $0$ zero-error capacities. 

Generalization of our analysis for channels with infinite alphabets and continuous channels is of research interests. The study of batched code design for a line network of channels like BSC is also desirable. 

Last, we are curious whether our outer bound holds without the batched code constraint.


\bibliographystyle{IEEEtran}
\bibliography{shortbib.bib}

\end{document}